\documentclass[11pt]{article}
\usepackage{amsthm}
\usepackage{fancyhdr}
\usepackage{fullpage}
\usepackage{times}

\usepackage{amssymb,amsmath}
\usepackage{amsmath}
\usepackage[final]{graphicx}
\usepackage{color}
\usepackage{xspace}

\usepackage{algorithmicx}
\usepackage[noend]{algpseudocode}
\usepackage{algorithm}
\usepackage{ifpdf}
\ifpdf    
\usepackage{hyperref}
\else    
\usepackage[hypertex]{hyperref}
\fi

\long\def\commabs #1\commabsend{}

\newcommand{\reals}{\hbox{$\rlap{\rm I} \> \kern-.2mm{\rm R}$}}

\newcommand{\etal}{{\em et al.\ }}

\newcommand{\e}{\epsilon}

\newtheorem{theorem}{Theorem}[section]
\newtheorem{lemma}[theorem]{Lemma}

\newtheorem{corollary}[theorem]{Corollary}

\newtheorem{definition}{Definition}

\newcommand{\comment}[1]{}

\def\QED{\quad\blackslug\lower 8.5pt\null\par}

\newcommand{\Remove}[1]{}
\newcommand{\remove}[1]{}

\begin{document}

\title{Another Proof of  Cuckoo hashing with New Variants}
\author{ Udi Wieder\\ { \small VMware Research}}
\maketitle

\begin{abstract}
	We show a new proof for the load of obtained by a Cuckoo Hashing data structure. Our proof is arguably simpler than previous proofs and allows for new generalizations. The proof first appeared in Pinkas \etal \cite{PSWW19} in the context of a protocol for private set intersection. We present it here separately to improve its readability.
\end{abstract}

\section{The Problem}
In the Cuckoo Hashing scheme we are presented with $n$ items $x_1,\ldots, x_n$ drawn from a finite universe $\mathcal U$. The goal is to store them in a hash table. To that end there are two arrays $A_0$, $A_1$, each with $m$ slots. A memory slot  can contain one item. We also have two hash functions $h_\sigma:\mathcal U \rightarrow [m]$, $\sigma\in\{0,1\}$, where each function takes an item $x_i$ and returns an index in $[m]$.  A \emph{legal placement} of the items is a mapping of each $x_i$ either to $A_0[h_0(x_i)]$ or to $A_1[h_1(x_i)]$, such that every slot is assigned at most one item. The problem we aim to solve is the following: assuming the hash functions are drawn uniformly at random from the set of functions $\mathcal U \rightarrow[m]$, how small could $m$ be while still guaranteeing the existence of a legal placement with high probability.
We will prove the following theorem:
\begin{theorem}\label{thm:mainbound}
	If the hash functions are drawn uniformly from the set of all functions $~\mathcal U \rightarrow [m]$ and $m \geq (1+\epsilon)n$ for  $\epsilon > 0$, then the probability there is a legal placement is $1-O(1/n)$ where the big $O$ notation hides constants depending on $\epsilon$.
\end{theorem}

\subsection{$d$-dimensional Cuckoo Hashing}
The motivation for coming up with a new proof stemmed from the need to find bounds for a generalization we call  $d$-dimensional cuckoo hashing. In this case each hash function maps an element from the universe to a $d$-dimensional vector in $[m]^d$. The goal is to place $d$ copies of an item, either in the $d$ locations in $A_0$ indexed by $h_0$, or in the $d$ locations in $A_1$ indexed by $h_1$. The case $d=1$ corresponds to standard Cuckoo Hashing.  Our proof below generalizes to this case and shows that it suffices to have  $m \geq (1+\e)d^2 n$. See \cite{PSWW19} for the original motivation and applications of this result. 
\subsection{Previous Approaches and Related Work}
Theorem~\ref{thm:mainbound} was first proved in \cite{PR04}, other proofs exist, see the survey  \cite{Wieder17}. We note that here we are concerned merely with the \emph{existence} of a placement while previous proofs bound the running time of an algorithm that \emph{finds}  a legal placement. It is not hard to see that our approach could be used to that end as well, however to the best of our knowledge it does not offer new insights or simplifications. We discuss that further in Section~\ref{sec:running-time}.
 
Previous proofs use the notion of a \emph{Cuckoo Graph}. A cuckoo graph is a  bipartite graph with $m$ vertices on each side. Each vertex represents a memory slot and edges represent items. One can show that a legal placement exists iff every connected component in the cuckoo graph has at most one cycle. There are several approaches for showing that. The  proof  in \cite{PR04} finds a small set of `forbidden graphs' and shows by enumeration that with high probability none of them appear. Alternatively, one can use techniques from random graph theory \cite{KMW09}. See survey \cite{Wieder17}.
 
We differ by looking at a different graph which we call the \emph{inference graph}. The inference graph tracks the logical constraints behind placement decisions, similar to inference graphs for 2-SAT. We show that a legal placement exists if and only if the inference graph does not contain a cycle.  We then proceed to show via  enumeration that a cycle is not likely to occur.

\section{The Inference Graph}
Given a pair of functions $h_\sigma: \mathcal U \rightarrow [m]$, $\sigma \in \{0,1\}$ and a set of items $S$, the \emph{Inference Graph} $G(S,h_0,h_1)$ is composed of the following: 
\begin{description}
	\item \emph{nodes: }The set of nodes is comprised of two sets $a^0_1,\ldots,a^0_n$ and
	$a^1_1,\ldots,a^1_n$. Semantically we think of node $a^\sigma_i$ as representing the event that $x_i$ was placed in $A_\sigma$.
	\item \emph{edges:} The edges of the graph are directed and represent inferences between the events: If
	$h_\sigma(x_i) = h_\sigma(x_j)$ then $G$ has the directed edges $(a^\sigma_i, a^{1-\sigma}_j)$ and $(a^\sigma_j,a^{1-\sigma}_i)$. In words,
	the edges $(a^0_i, a^1_j)$, $(a^0_j, a^i_1)$ mean that we cannot place both $x_i,x_j$ in $A_\sigma$: if $x_i$ is placed in $A_\sigma$ then $x_j$ must be placed in $A_{1-\sigma}$ and vice versa.
\end{description}

Note that this graph is somewhat similar to the structure constructed for resolution proofs of $2-$SAT formulas.

\subsection{Placement}
The goal of the following definitions is to form conditions under which an item $x_i$ could be placed. We then
will show that these conditions hold w.h.p for all items.  Denote by $G(a^\sigma_i)$ the set of vertices reachable
from $a^\sigma_i$ (including $a^\sigma_i$) in the inference graph.  We may drop the $\sigma$ in our notation as our claims hold for both $\sigma = 0$ and $\sigma = 1$. When we refer to the items of $G(a_i)$, we mean all items associated with nodes of $G(a_i)$, that is: $\{ x_j : a_j \in G(a_i)\}$.

The first observation to make is that if a node $a^\gamma_{j} \in G(a^\sigma_{i})$ then any legal  placement in which $x_{i}$ is placed in table $A_\sigma$ must have $x_{j}$ placed in  $A_\gamma$. 

\begin{definition}
	A node $a_i^\sigma$ in the Inference Graph is called \emph{bad} if there is a $j$ such that both $a^0_j, a^1_j \in G(a_i^\sigma)$. An item $x_i$ is \emph{bad} if both $a^0_i$ and $a^1_i$ are bad. 
\end{definition}
Namely, a node $a^\sigma_{i}$ is bad if there is an item $x_{j}$ such that placing $x_{i}$ in table $A_\sigma$ prevents
placing $x_{j}$ in either table $A_0$ or table $A_1$. So there is no legal placement in which $x_i$ is placed in $A_\sigma$.
An item $x_{i}$ is bad if placing it in either table $A_0,A_1$ prevents finding a placement for other items.

Clearly, a bad item implies that not all items could be placed. The following lemma states the converse is also true.
\begin{lemma}\label{lem:basicplacement}
	If node $a^\sigma_i$ is not bad then all items of $G(a^\sigma_i)$ could be placed.
\end{lemma}
\begin{proof}
	We first place $x_i$ in $A_\sigma$. Then place all its neighbors in $A_{1-\sigma}$ and continue iteratively. Note that a node associated with an occupied slot is part of $G(a_i)$. Now, if  item $x_j$ cannot be placed then it must intersect items both on $A_0$ and on $A_1$ which means both $a^0_j$ and $a^1_j$ are in $G(a_i)$ which is a contradiction. 
\end{proof}
\begin{lemma}
	If none of the items are bad then all items could be placed in the tables.
\end{lemma}

\begin{proof}
	The algorithm that places all the items is now straightforward: Let $S$ be the set of currently unplaced items. Pick an item $x_i \in S$ and since it is not bad, then $a^\sigma_i$ is not bad for some $\sigma \in \{0,1\}$. Now by Lemma~\ref{lem:basicplacement} all  items of $G(a_i)$ could be placed successfully.  Let $S'$ be the remaining items, i.e., $S' = S \setminus G(x_i)$.  Given $S'$ and all the free locations in $A_0,A_1$ we compute the new inference graph $G' = G(h_0,h_1,S')$ and continue inductively. The only thing remaining to observe is that if there were no bad items in $G$ then there are no bad items in $G'$. To see this observe that $G'$ is a subgraph of $G$. Indeed let $x_j$ be an item in $G'$. Note that both slots $h_\sigma(x_j)$ must be free, otherwise $a^\sigma\in G(a_i)$, so every inference made in $G'$ is true also for $G$.
\end{proof}


\subsection{Main Result}
The goal now is to calculate the probability an item is bad. 
\begin{theorem}\label{thm:main}
	If  the size of each table is greater than $m=(1+\e)n$ then for every item $i$, the probability $x_i$ is bad is at most $\left(\frac{1+\e}{\e}\right)^3\cdot \frac{2}{ m^2}$,  where the probability is taken over the
	choice of the hash functions.
	\label{theorem1}
\end{theorem}
Taking a union bound over all $x_i$ proves Theorem~\ref{thm:mainbound}. 
The remainder of the section is dedicated to the proof Theorem~\ref{thm:main}.

Our approach is to show that it is unlikely that an item is bad. For that to happen both its nodes need to be bad, and we would like to count how many bad graphs are there and show that they are unlikely to appear. As is often the case in proofs based on counting argument, the trick is to carefully define the objects which we count. In order to facilitate this bound we need to constrain further the exact notion of a bad node, captured by the next definition:
\begin{definition}
	A \emph{bad path} rooted at $a^\sigma_i$ is a simple path from $a^\sigma_i$ to $a^{1-\sigma}_i$. A bad path is called \emph{basic} if it does not contain a bad path. In other words, for each $j \neq i$ at most one of $\{a^0_j,a^1_j\}$ can appear in the path. 
\end{definition}

The next lemma shows that basic bad paths are the only type of subgraphs we need to care about.

\begin{lemma}
	If a node is bad then it is the root of a basic bad path. 
\end{lemma}
\begin{proof}
Assume $a^\sigma_i$  is bad, there must be at least one $j$ for which $a^\sigma_i$ is connected to both $a^0_j, a^1_j$. Further, we can assume that there is no $k \neq j$ such that both $a^0_k$ and $a^1_k$ appear on the paths from $a^\sigma_i$ to $a^0_j, a^1_j$. We can make this assumption because if there is, we may take the pair $a^0_k,a^1_k$ instead. 

Now recall that by construction, if an edge $(a^\sigma_k, a^{1-\sigma}_\ell)$ appears in the inference graph, then so does the edge $(a^\sigma_\ell, a^{1-\sigma}_k)$. A simple induction shows that if there is a path $a^\sigma_i \rightsquigarrow a^{1-\sigma}_j$ then there is a path $a^\sigma_j \rightsquigarrow a^{1-\sigma}_i$. Thus, we can construct a path  $a^\sigma_i \rightsquigarrow a^{\sigma}_j \rightsquigarrow a^{1-\sigma}_i$. Further, there is no item which is repeated in the path. 
\end{proof}

We now need to show that the probability an item is bad is small. We start by bounding the probability a node is the root of a basic bad path.
Recall that $m$ denotes the number of slots in each array. 
\begin{lemma}
	If $m\geq (1+\e)n$ then for every item $i$, the probability  $a^0_i$  is the root of a basic bad path is at most $\frac{1+\e}{\e m}$.  
\end{lemma}

\begin{proof}
We count the number of labeled  basic bad paths rooted at $a^0_i$.
Let $k$ be the number of edges in the path, and $k+1$ the number of nodes. The head of the path is already set to be $a^0_i$, and the tail is $a^1_i$ so there are at most $n^{k-1}$ ways of choosing the nodes in between. We conclude: 
\begin{align}\label{eq:numberpaths}
	\#\text{possible labeled basic bad paths rooted at }a^0_i \leq 
	 \sum_{k}n^{k-1}  
\end{align}

Given a labeling of a bad path, we calculate the probability it actually appears in the graph. Consider an edge in the path $(a^\sigma_k,a^{1-\sigma}_\ell)$. This edge belongs to the inference graph iff $h_\sigma(x_k) = h_\sigma(x_\ell)$  which happens\footnote{This is the only place we need to augment the proof for the $d$-dimensional case, where the probability is roughly $d^2/m$} with probability $\leq 1/m$.   Now, since no item in the path is repeated, the occurrences of the $k$ edges are independent events. The probability a given path appears in the graph is therefore $\leq 1/m^k$. 

Combined with \eqref{eq:numberpaths} we have that the probability $a^0_i$ is the root of a bad path is at most 
\begin{align}\label{eq:pathprob}
\frac{1}{m}\sum_{k\geq1} \left(\frac{n}{m}\right)^{k-1} \leq \frac{1}{m}\sum_{k\geq1} \left(\frac{1}{1+\e}\right)^{k-1} \leq \frac{1+\e}{\e m}
\end{align}
\end{proof}

\begin{proof}[Proof of Theorem \ref{thm:main}]
Note that we bounded the probability node $a^0_i$ is bad. For item $x_i$ to be bad node $a^1_i$ has to be bad as well. Again, it is enough to bound the case $a^1_i$ is the root of a basic bad path. So we condition on $a^0_i$ having a basic bad path. When accounting for all possible basic bad paths rooted at $a^1_i$ we need to differentiate between the various ways in which they intersect the basic bad path rooted in $a^0_i$. We proceed via a small case analysis.

\paragraph{Case $1$:} Assume the items associated with the bad path do not intersect those of the path from $a^0_i$. In this case the conditioning on the path from $a^0_i$ has no affect; the same calculation holds as before and the probability item $x_i$ is bad is at most $\left(\frac{1+\e}{\e m}\right)^2$. 
\paragraph{Case $2$:}
There is some item $x_j$ which belongs to both bad paths rooted at $a^0_i$ and $a^1_i$. Let $k_1$ be the length of the bad path starting at $a^0_i$ and $k_2$ be the length of the prefix of the bad path starting at $a^1_i$ and ends at the intersection: $a^\sigma_j$. Note that there are at most $k_1$ possibilities for choosing $x_j$. The probability such a path exists is therefore at most
\begin{align*}
&\frac{1}{m}\sum_{k_1\geq 1}\left(\frac{1}{1+\e}\right)^{k_1-1}\frac{2k_1}{m}\sum_{k_2\geq 1}\left(\frac{1}{1+\e}\right)^{k_2-1} \\
 \leq &\frac{2(1+\e)}{\e m^2} \sum_{k_1\geq 1}k_1\left(\frac{1}{1+\e}\right)^{k_1-1}\\
 \leq & \left(\frac{1+\e}{\e}\right)^3\cdot \frac{2}{ m^2}
\end{align*}
\end{proof}
Taking a union bound over all items we have:
\begin{corollary}
	If $m\geq (1+\e)n$ then the probability there is a failure is at most $ \frac{2(1+\e)^2}{\e^3}\cdot \frac{1}{n}$.
\end{corollary}

\section{Remarks}

\subsection{Generalization to multi-dimensional cuckoo hasing} 
What is there to gain from a new proof?  it is our subjective view that this proof is simpler and easier to follow, at least as an offline result. However, the main use of a new proof to a known result is to increase the understanding of the result and hopefully generalize in a new direction. Indeed the motivation for this work stemmed from a new and different application of a Cuckoo Hashing variant used in protocols for private set intersection. In this variant which we call $d$-dimensional cuckoo hashing, each hash function maps an element from the universe to a $d$-dimensional vector in $[m]^d$. The goal is to place $d$ copies of an item, either in the $d$ locations in $A_0$ indexed by $h_0$, or in the $d$ locations in $A_1$ indexed by $h_1$. The case $d=1$ corresponds to standard Cuckoo Hashing.  It is straightforward to see that the exact same proof holds, the only difference being in equation~\ref{eq:pathprob} where now the probability an edge appears in the path is $d^2/m$. See \cite{PSWW19} for the original motivation and application of this result. 

\subsection{Online vs Offline and Running time}\label{sec:running-time}
Note that in this manuscript the problem tackled is the existence if a legal placement, without an explicit argument to find it. Of course, the proof in effect describes an insertion algorithm but does not argue anything about its running time and further, assumes that the entire graph is given in advance.  A major advantage of Cuckoo Hashing however is that the dynamic insertion algorithm, as described in \cite{PR04} is efficient, namely takes $O(1)$ on expectation and $O(\log n)$ w.h.p.. In previous proofs the running time of the insertion algorithm is analyzed via a careful analysis of the structure of the cuckoo graph, in particular, the running time is bounded via a bound on the size of its connected component. A similar type of argument could be made here for the inference graph as well. It offers no new simplifications or generalization (to the best of our understanding) so we saw no point in spelling it out.

\subsection{Stash}
In \cite{KMW09} it is shown that augmenting the scheme with $s$ extra slots that can hold any item (a.k.a stash) reduces the probability of insertion failure to roughly $n^{-(s+1)}$. A more combinatorial proof was given in \cite{Aumuller2012}. See also Theorem $5.5$ in \cite{Wieder17}.  A similar argument could be made here, but we do not think it offers simplifications or new insights.

\subsection*{Acknowledgments}
I'm indebted to my co-authors of \cite{PSWW19}:  Thomas Schneider, Christian Weinert, and especially Benny Pinkas.

\bibliographystyle{alpha}
\bibliography{refnewproof}
\end{document}